\documentclass{llncs}

\usepackage{graphicx}
\usepackage{amsmath}
\usepackage{algorithmic}
\title{Line Formation by Fat Robots under Limited Visibility}

\author{Arijit Sil$^1$ \and Sruti Gan Chaudhuri$^2$}

\institute{$^1$ Meghnad Saha Institute of Technology, Kolkata, India\\ $^2$ Jadavpur University, Kolkata, India}

\begin{document} 

\maketitle              

\begin{abstract}
This paper proposes a distributed algorithm for a set of tiny unit disc shaped robot to form a straight line. The robots are homogeneous, autonomous, anonymous. They observe their surrounding up to a certain distance, compute destinations to move to and move there. They do not have any explicit message sending or receiving capability. They forget their past observed or computed data. The robots do not have any  global coordinate system or origin. Each robot considers its position as its origin. However, they agree on the X and Y axis. The robots are not aware of the total number of robots in the system. The algorithm presented in this paper assures collision free movements of the robots.  To the best of our knowledge this paper is the first reported result on line formation by fat robots under limited visibility. 
\end{abstract}

\section{Introduction} 
One of the current trends of research in the field of Robotics is to replace a big robot by a group of small autonomous robots who work in coordination between themselves. The group of robots may perform many real time jobs like moving a big object, cleaning a big surface, guarding a geographical area etc. In theoretical point of view, one of the fundamental tasks for executing this kind of jobs is to form the geometric patterns on the plane by the robots. 
In this paper we address the problem of line formation by unit disc shaped robots or {\it fat robots}. 
The robots can sense their surrounding up to a certain range. They compute their destination locations by our proposed algorithm and move there. After reaching their destinations they forget all their past sensed and computed data. 

All reported line formation algorithms \cite{FPS2012book} for mobile robots considers that the robots as points and they are able to sense all other robots. 
A point robot neither creates any visual obstruction nor acts as an obstacle in the path of other robots. Czyzowicz et. al,\cite{Reference14} extended the traditional weak model of robots by replacing the point robots with unit disc robots (fat robots). Only some solutions on gathering problem has been reported for fat robots \cite{AGM2013PODC,GM2015JDA,GM2016book}. Under limited visibility gathering is solved for point robots \cite{FPSW2005TCS} and fat robots \cite{BKF2012SE}.
Dutta et. al \cite{Reference15} proposed a circle formation algorithm for fat robots assuming common origin and axes for the robots. Here the robots are assumed to be transparent in order to avoid visibility block. However, a robot acts as an physical obstacle if it falls in the path of other robots. The visibility range/radius of the robots is assumed to be limited.  
Datta et. al\cite{Reference16}, proposed another distributed algorithm for circle formation by a system of mobile asynchronous transparent fat robots with unlimited visibility. 

 In this paper we consider fat robots and propose a collision free movement strategy to form a line where the robots can only sense other robots up to  a finite distance.

\subsection{Underlying Model}

The robot model used in this paper is describes as follows. 

\begin{itemize}

\item The robots are autonomous.

\item Robots are anonymous and homogeneous i.e., they are not uniquely identifiable.

\item	 A robot is represented as a transparent disc with unit radius. The robots are transparent or see-through in order to ensure full visibility, but they act as physical obstructions for other robots.

\item	The robots do not have any global coordinate system. Each robot considers its position as its origin. They agree on the direction of $XY$ axes. The robots also can agree on unit distance (the radius of the robots can be considered as unit).

\item	 Every robot executes a cycle of three phases:
\begin{itemize}
\item {\bf Look} - the robot takes a snapshot around itself up to a finite range and identifies the other robots' positions w.r.t its own coordinate system; 
\item {\bf Compute} - based on other robot positions, the observer robot computes its destination; 
\item {\bf Move} - the robot moves to the destination point calculated in the previous phase. 
\end{itemize} 

\item The robots execute this cycle in semi-synchronous scheduling where an arbitrary set of robots look, compute and move simultaneously. This scheduling assures that when a robot is moving no other robot is observing it. 

\item The robots do not stop before reaching its destination ({\it rigid} motion).

\item	The robots are oblivious in the sense that they cannot remember any data from previous cycle.

\item A robot can see up to a fixed distance around itself comprised of a circular area centered at the center of the robot having radius length $rad_{v}$ on the 2D plane. 

\item The robots do not know about the total number of robots in the system.

\item	Robots cannot communicate using explicit messages. 
 
\item	The robots form a graph $G \big(V, E\big)$. Every robot is a vertex $v$ in $G$, where $v \in V$ of graph $G$. There exists an edge $e$, $e \in E$ between the robots $r_{i}$ and $r_{j}$ if and only if they can see each other. Initially, the graph $G$ is assumed to be connected which implies that every robot can see at least one other robot. 

\item Initially the robots are stationary and the mutual distance between two robots is atleast $\delta > 2 units$.

	
	
	
	         
	          
	              
	
	\end{itemize}

\section{Overview of the problem}
Let $R$ be a set of $n$ robots under the model described in previous section. The robots are assumed to be transparent in order to ensure full visibility, but they act as physical obstacles for other robots. A robot is represented by its center, i.e., by $r$ we mean a robot whose center is $r$. The robots in $R$ have to move in such a way that after a finite number of execution cycles, the robots in $R$ will form a straight line.    

When a robot $r\in R$ becomes active it first enters into look state. In this state $r$ takes a snapshot of the robots that are present within its visibility circle and plots those robots in its local coordinate system $Z_{r}$. This set of robots visible to $r$ is called the neighbours of $r$. With respect to $Z_{r}$, the set of robots visible to $r$ can be divided into eight distinct and non-overlapping sets.
 Refer to Fig. \ref{fig1}, 
 \begin{itemize}
 \item Set $A$ consists of the robots partially or fully present in the area of unit distance around  the positive $Y$ axis but it does not contain robot $r$ itself.
 \item Set $C$ consists of the robots partially or fully present in the area of unit distance around the positive $X$ axis.
 \item The set $F$ represents the robots partially or fully present in the area of unit distance around the negative $X$ axis and set $I$ contains the robots partially or fully present in the area of unit distance around negative $Y$ axis. 
 \item Set $B$ and $D$ account for the robots that fall within the first and second quadrant of the local coordinate system respectively (leaving the robots in $A$, $C$ and $I$). 
 \item The robots which are in third and fourth quadrant of the coordinate system $Z_{r}$ make the sets $E$ and $G$ respectively (leaving the robots on $A$, $F$ and $I$). 
 \end{itemize}

\begin{figure}[h]
\centering  
	\includegraphics[scale = 0.8]{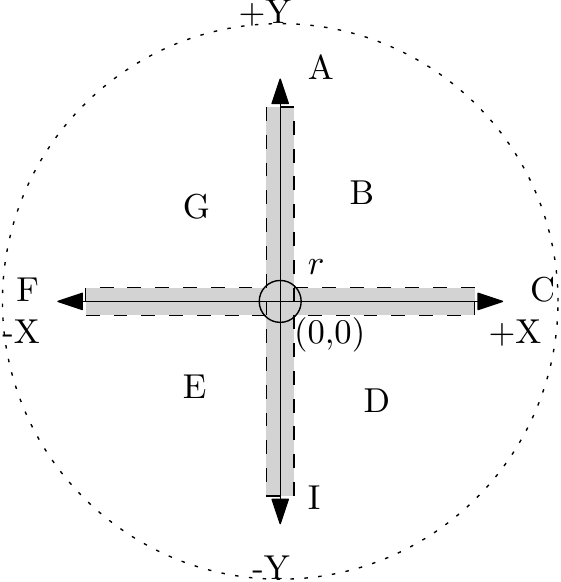}
\caption{Local view of a robot}
\label{fig1}
\end{figure}

In the first step of the computation phase $r$ calls a routine to check whether it should make a move in this cycle or not. The routine considers all possible scenarios and determines a destination point for $r$.

In the second step of the computation phase $r$ calls a second routine that returns the amount of horizontal and /or vertical shift that is allowed for $r$ in accordance with within $r$'s visibility circle.

In the final phase of the current cycle $r$ moves to the destination point computed in the previous phase. The movement ensures that it preserves connectivity and avoids collision with the neighbouring robots. 

\paragraph{\bf Vacant Point:}
A point is vacant if there exist no parts of another robot around a circular region of radius 1 around this point.

\paragraph{\bf Free path:}
A path of a robot is called free path, if from source to destination point (Refer to Fig. \ref{freepath}) the rectangular area having length as the source to destination distance and width as two units, is not contained any part of another robot. 
 \begin{figure}[h]
   \centering
   \includegraphics[scale=0.8]{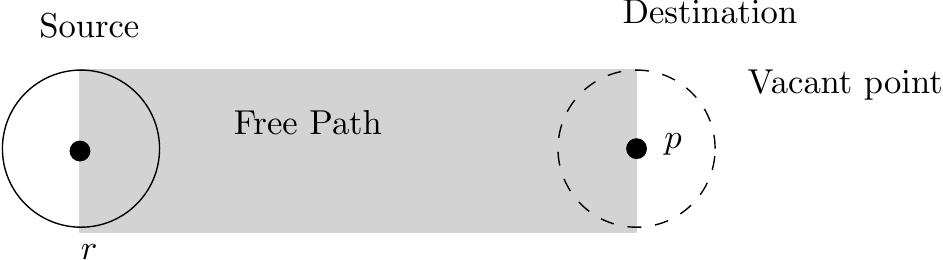}
   \caption{An example of free path of robot $r$ and Vacant point $p$}
   \label{freepath}
  \end{figure}

\section{Description of the Algorithm}

 The algorithm we have developed is executed by the robots  in semi- synchronous manner in computation phase and determines the distance and direction of their movement. The algorithm $LineForm()$ is divided into two subroutines. One subroutine, $NoMovement()$ identifies the cases where the robot will not move. whereas the other subroutine $getDestination()$ computes the destination points of the robots. Finally the robots move to this computed destinations
  
\paragraph{\bf NoMovement(r):}
$r$ will not move for the following configurations.
\begin{itemize}
\item If, $r$ does not find any other robot in its visibility circle \footnote{Which happens if and only if there is only one robot.}(Fig. \ref{nomove}(a)).
\item If $F$ is not empty (Fig \ref{nomove}(b)).
\item If sets $E \cup F \cup G$ are not empty (Fig \ref{nomove}(c)).
\item If $B \cup C \cup D \cup E \cup F \cup G$ is empty and $A \cup I$ is not empty. (Fig \ref{nomove}(d)).

\end{itemize}

\begin{figure}[h]
\centering  
	\includegraphics[scale = 0.8]{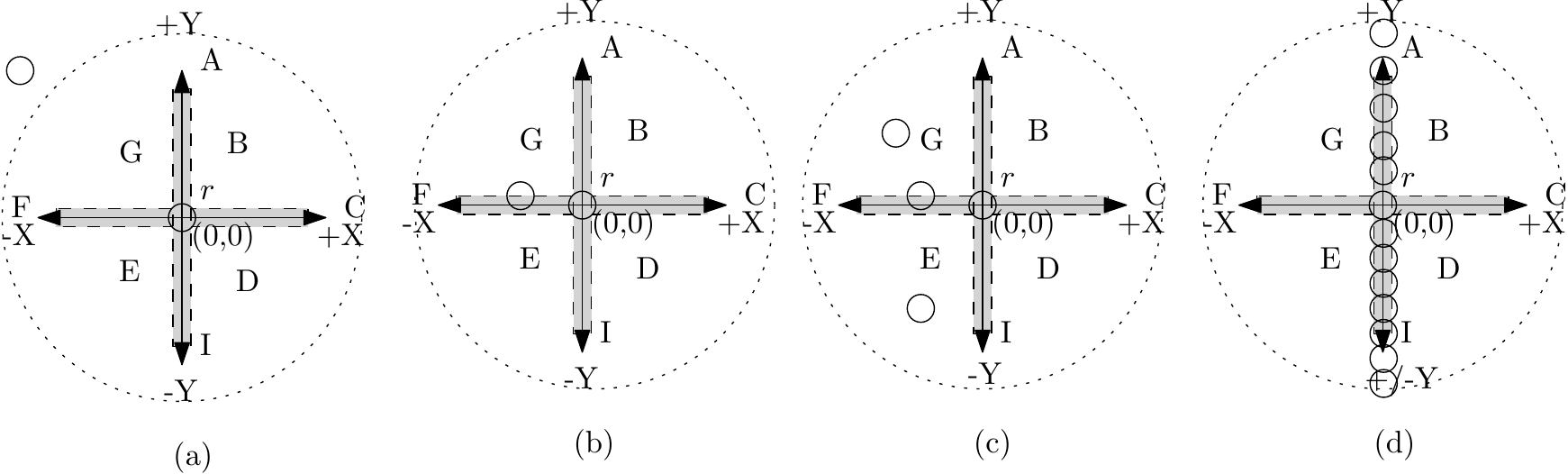}
\caption{No movement configurations for $r$}
\label{nomove}
\end{figure}

\vspace{-20pt}
\paragraph{\bf getDestination(r):} This subroutine considers all the robots that are visible to $r$ and ahead of it in the direction of positive $X$ axis (hereinafter referred as RIGHT). If $E \cup F \cup G$ is not empty the robot $r$ does not move. Hence, this routine takes into account only the robots in the sets $B$, $C$ and $D$. If $B \cup C \cup D$ is empty then the subroutine returns the current location of $r$ as the destination point. But if $B \cup C \cup D$ is not empty then the algorithm finds the nearest axis vertical to $X$ axis that contains one or more robots from $B \cup C \cup D$. Let, $\psi_{right}$  be that axis. It then considers four different scenarios.

Suppose, the robots partially of fully present in the area of unit distance around  the axis $\psi_{right}$ that are within the visibility range of $r$ forms the set $R_{\psi_{right}}$. The coordinate of $r$ is taken as $(0,0)$ w.r.t. its local coordinate system.

\begin{figure}[h!]
\centering  
	\includegraphics[scale = 0.8]{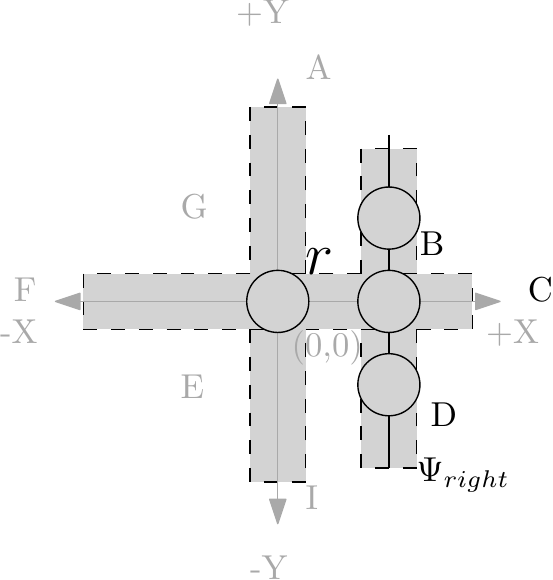}
\caption{An example of set $R_{\psi_{right}}$}
\label{fig2}
\end{figure}

\newpage
{\bf Scenario 1:} If $R_{\psi_{right}} \cap C$ is empty,  (Fig. \ref{sc1}), which means none of the robots in $R_{\psi_{right}}$ resides on $X$ axis, then the coordinates of the destination point $p$ is given by: 

           $x_{p} = 0 + \Delta x$ [where $\Delta x$ = offset along $X$ axis from current position of $r$ to the intersection point between $X$ axis and $\psi_{right}$].
          
           $y_{p} = 0$.
           
\begin{figure}[h]
\centering  
	\includegraphics[scale = 0.7]{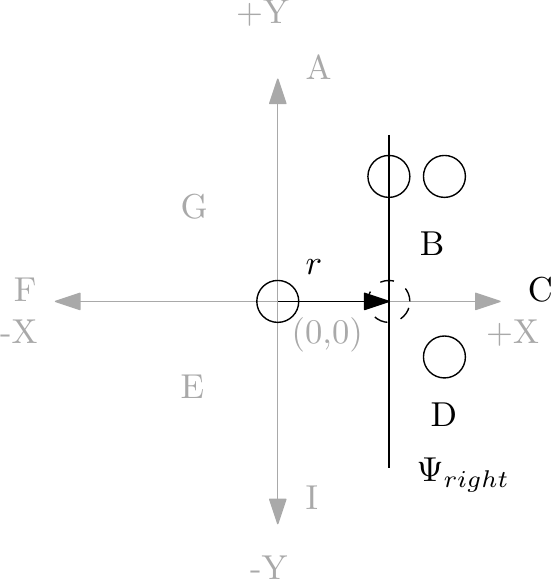}
\caption{Scenario 1}
\label{sc1}
\end{figure}

{\bf Scenario 2:} If $R_{\psi_{right}} \cap C$ is not empty and $R_{\psi_{right}} \cap B$ is not empty and$R_{\psi_{right}} \cap D$ is empty then $R_{\psi_{right}}$ has a robot on $X$ axis, it has got one or more robot that belongs to set $B$, but does not have any robot that belongs to set $D$. So, if $r$ only moves horizontally it will collide with the robot already present on the intersection point of $\psi_{right}$ axis and $X$ axis. So, its horizontal movement must be followed by a vertical movement towards the negative $Y$ direction. As the set $D$ does not contain any robot on $\psi_{right}$ axis, $r$ does not have to face collision with any other robot as it moves vertically. So the coordinates of the destination point $p$ is given by:

$x_{p} = 0 + \Delta x$ 

$y_{p} = 0 - \Delta y$ [where $\Delta y$ = the radius of robots]

If $p$ is not a vacant point, the subroutine re-computes $y_p$ as,

 $y_{p} = 0 - i\Delta y$ [where $\Delta y$ = the radius of robots, i=2,3..]. 
 
This process continues till a vacant point is found.
 
 If no vacant point is found in $R_{\psi_{Right}}\cap B$ using this procedure, then the vacant point can be found in the same manner in $R_{\psi_{Right}}\cap D$.

\begin{figure}[h]
\centering  
	\includegraphics[scale = 0.8]{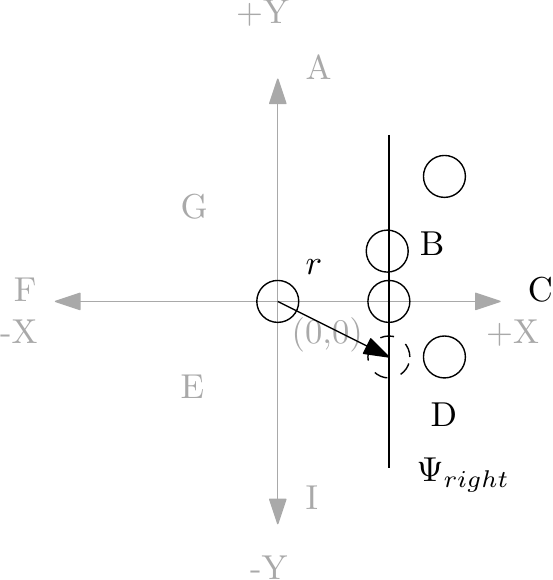}
\caption{Scenario 2}
\label{sc2}
\end{figure}

{\bf Scenario 3:} If $R_{\psi_{right}} \cap C$ is not empty and $R_{\psi_{right}} \cap B$ is empty and $R_{\psi_{right}} \cap D$ is not empty then $R_{\psi_{right}}$ has a robot on $X$ axis, it has no robot that belongs to set $B$, but does have one or more robot that belongs to set $D$. So, if $r$ only moves horizontally, just as in the previous scenario, it will collide with the robot already present on the intersection point of $R_{\psi_{right}}$ axis and $X$ axis. So, its horizontal movement must be followed by a vertical movement towards the positive $Y$ direction. 

So the coordinates of the destination point $p$ is given by:

$x_{p} = 0 + \Delta x$ 

$y_{p} = 0 + \Delta y$ [where $\Delta y$ = the radius of robots]

If $p$ is not a vacant point, the subroutine re-computes $y_p$ as,

 $y_{p} = 0 + i\Delta y$ [where $\Delta y$ = the radius of robots, i=2,3..]. 
 
This process continues till a vacant point is found.
 
 If no vacant point is found in $R_{\psi_{Right}}\cap D$ using this procedure, then the vacant point can be found in the same manner in $R_{\psi_{Right}}\cap B$.

\begin{figure}[h]
\centering  
	\includegraphics[scale = 1]{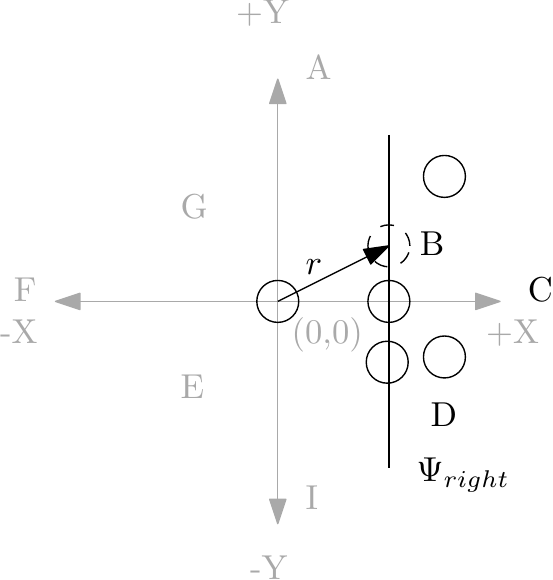}
\caption{Scenario 3}
\label{sc3}
\end{figure}

{\bf Scenario 4:} If there is no vacant point in $R_{\psi_{right}} \cap B$ and $R_{\psi_{right}} \cap C$ and $R_{\psi_{right}} \cap B$, then the robot $r$ computes its destination point as follows.
\begin{itemize}

\item If $I$ is completely empty, then $r$ moves towards either $-Y$ axis such that, 

$x_p = 0$.

$y_p = 0 - rad_v +1$.
 
\item If $I$ is not empty, then $r$ moves along $-Y$ axis such that,

$x_p = 0$.

$y_p = 0 - d+1$  where $d=$ the vertical distance with $r$ and the robot with maximum $y$ value in $I$. 

\end{itemize}

\begin{figure}[h]
\centering  
	\includegraphics[scale = 1]{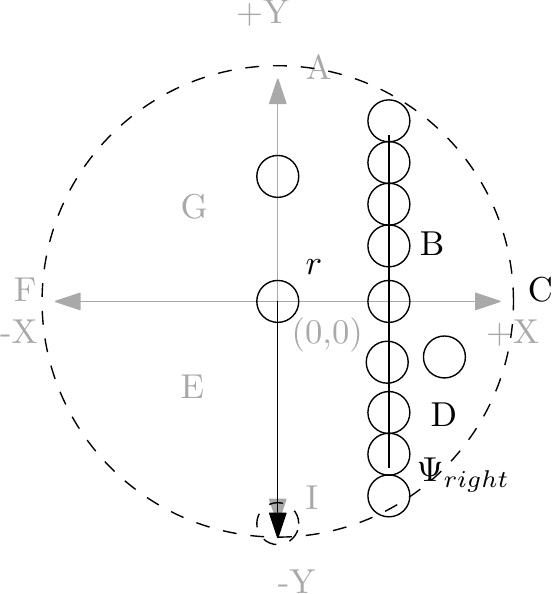}
\caption{Scenario 4}
\label{sc4}
\end{figure}

\section{Correctness}
The robots successfully form a straight line in finite time using our proposed algorithm. The algorithm is correct as it gives assurance of the following facts. 
\begin{itemize}
\item The visibility graph $G$ does not becomes disconnected.
\item The robots do not collide due to their movement strategies.
\item The robots do not fall into deadlock and form the line in finite time.
\end{itemize} 

Following lemmas are presented to prove these facts.
\begin{lemma}
\label{Lemma 1}
 The connectivity graph $G$ remains connected.
\end{lemma}
\begin{proof}
Consider scenario 1. The robots in $B$ or $C$ or $D$ will not move due to the presence of $r$ according to our $NoMovement()$ subroutine. Now we will show that when $r$ is moving to its destination it gets closer to the robots in $B$ or $C$ or $D$. Without loss of generality let us prove this by taking the existence of any robot in $B$. The same arguments hold for the presence of robots in $B$ or $C$ or $D$.

\begin{figure}[h]
\centering  
	\includegraphics[scale = 1]{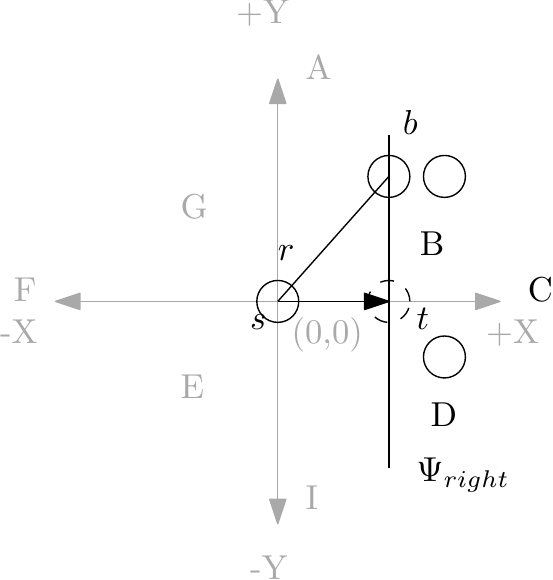}
\caption{An example of scenario 1 and connection preservation}
\label{cor1}
\end{figure}

Let $s$ be the starting location of $r$ (Fig. \ref{cor1}). Let there exists a robot $r_b$ in $B$ at $b$. Let $t$ be the destination of $r$. $r$ moves along the edge $st$ of the triangle $stb$. In triangle $sbt$ $|tb|<|sb|$ (since $sb$ is the diagonal.) Hence, when $r$ moves to $t$, it becomes closer to the robot present in $B$. 

Using the similar argument we can prove that under scenario 2 and 3, due to the movement of $r$, it becomes closer to the other robots present in $B$, $C$ and $D$, Hence, there is no chance to get disconnected with any robot. 

Now consider Scenario 4. No robots in $B$ or $C$ or $D$ moves following $No$ $Movement()$ subroutine. First consider the case when $I$ are empty. Then $r$ moves along $-Y$, $rad_v$ distance. Note that $r$ in its new position is connected with the robots in $D$. Note that the robots in $D$ in connected with $B$, $C$, $D$, $A$. Hence $G$ remains connected. Now suppose $I$ is not empty. $r$ moves towards $-Y$ till it touches the robots in $I$. Due to this movement $G$ does not become disconnected.
\qed
\end{proof}

\begin{lemma}
\label{Lemma 4} The robots never collide.
\end{lemma}
\begin{proof}
In scenario 1, 2 and 3 the robot $r$ moves to its next visible vertical line $\psi_{right}$ only when there is a {\it vacant point} on it. 

Consider scenario 1, the $r$ moves to the intersection point of $+X$ axis and $\psi_{right}$. This point is vacant according to the algorithm. The path towards this point from $r$ is also a it {\it free path} as there is no vertical line between $y$ axis and $\psi_{right}$.

Consider scenario 2, the $r$ moves to the vacant point on $\psi_{right}$ at $B$. The path towards this point from $r$ is also a free path as there is no vertical line between $y$ axis and $\psi_{right}$.

Consider scenario 3, the $r$ moves to the vacant point on $\psi_{right}$ at $D$. The path towards this point from $r$ is also a free path as there is no vertical line between $y$ axis and $\psi_{right}$. 

In scenario 1,2 and 3 no robot in the visibility circle moves other than $r$. $r$ moves in such a way that it does not collide with any other robot.

Consider scenario 4, the $r$ moves down along $-Y$ axis if it has free path. Otherwise it does not move. Hence, there is no chance for collision.

Note that the destination point for $r$ is chosen in such a way that no robot which is not visible to $r$, can come in $r$'s path.

Hence for all movements of $r$, it does not collide with any other robot. 
 \qed
\end{proof}

\begin{lemma}
\label{Lemma 2}
There exists always a robot which will move unless the robots in $R$ forms a straight line.
\end{lemma}
\begin{proof}
If a robot see any robot at its right side it will move. If it does not see any robot at $B \cup C \cup D$, it does not move. This is possible for following two cases.
\begin{itemize}
\item There is a single robot.
\item The robots have formed a straight line.
\item The robots do not form a straight line but there exist another robot in $A$ or $I$ which has the connectivity with the right side or left side of the $Y$ axis. For both the cases there exists robot other than $r$, which will move.
\end{itemize}
Hence, the lemma is true.
\qed
\end{proof}

\begin{lemma}
If the robots do not form a line yet it will leave its $Y$ axis after a finite time and move in the $+X$ direction.
\end{lemma}
\begin{proof}
Follows from lemma \ref{Lemma 2}.
\qed
\end{proof}

Given a set of robots $R$ on the 2D plane in its initial configuration, we may assume the existence of a line passing through the global right most robots. Let this line be the Right Most Axis (RMA) of initial configuration. In fact through our algorithm the robots finally are placed on the RMA and form the required straight line.

\begin{lemma}
\label{movetosir}
Each robot will move closer to RMA in finite time interval.
\end{lemma}

\begin{proof}
If a robot finds any robot at its right side it moves in $+X$ direction following scenarios  1, 2, 3 and moves to its $\psi_{right}$. 
As the robots do not stay idle for infinite time, it is guaranteed that the robots will reach to its $\psi_{right}$ in finite time, i.e., closer to RMA. 

In scenario 4, $r$ moves down along $-Y$ axis if there is a free path. When $r$ moves down its visibility circle also moves down and it covers a new set of robots. Eventually $r$ moves towards right and placed on the next vertical line nearer to RMA. If $r$ does not move down, there exists another robot to move and eventually $r$ gets its chance to move unless the straight line is already formed.  
\qed
\end{proof}

\begin{lemma}
\label{Lemma 3} None of the robots ever crosses the RMA of the set of robot.
\end{lemma}
\begin{proof}
Suppose there is a robot $r$ which has crossed RMA. $r$ can do that in two different ways. If $r$ was initially on RMA then it has left that axis to move to RIGHT or $r$ was initially LEFT of RMA and has crossed it while going towards RIGHT. In the first case to leave RMA $r$ has to observe an axis containing robots towards RIGHT. But as $r$ was sitting on RMA no such axis can exist and therefore it contradicts our assumption and therefore once on RMA, $r$ cannot leave it anymore. In the second case $r$ was LEFT of RMA and in order to go past RMA it has to find RMA to be the nearest axis containing robots towards RIGHT as there is no other axis with robots RIGHT of RMA. But if RMA is the nearest axis then the maximum horizontal shift would take $r$ up to RMA and not beyond that and once it reaches RMA, $r$ cannot leave it anymore. Therefore, the scenario contradicts our assumption. So, by contradiction we can say that none of the robots ever crosses RMA.
\qed
\end{proof} 

\begin{lemma}
\label{finite}
All the robots in $R$ will be on the RMA in finite time.
\end{lemma}
\begin{proof}
According to lemma \ref{movetosir}, each robot reach its $\psi_{right}$ in finite time. This implies that after a finite time there will be configuration when RMA will be the $\psi_{right}$ for each robot. After this configuration in finite time all robots will move to RMA. 
\qed
\end{proof}

 \section{Conclusion}
       Finally we can summarize the result in the following theorem
       \begin{theorem}
       A set of asynchronous, oblivious fat robots can form a straight line under limited visibility and one axis agreement without collision.
       \end{theorem}
      The future scope of this work would be to find the possibility of other pattern formation by fat robots or study the same problem removing the axis agreement or placing the robots uniformly distributed on the line.

\end{document}